\documentclass[journal]{IEEEtran}
\ifCLASSINFOpdf
\else
\fi

\usepackage{amsmath}
\usepackage{amssymb}
\usepackage{calligra}
\usepackage{amsthm}
\usepackage{algorithm}
\usepackage[noend]{algpseudocode}
\usepackage{graphicx}
\usepackage{epstopdf}
\usepackage[noadjust]{cite}
\usepackage{booktabs}

\usepackage{cite}
\usepackage{picinpar}
\usepackage{amsmath}
\usepackage{url}
\usepackage{flushend}
\usepackage[latin1]{inputenc}
\usepackage{colortbl}
\usepackage{soul}
\usepackage{multirow}
\usepackage{pifont}
\usepackage{color}
\usepackage{alltt}
\usepackage[hidelinks]{hyperref}
\usepackage{enumerate}
\usepackage{siunitx}
\usepackage{breakurl}
\usepackage{epstopdf}
\usepackage{pbox}

\usepackage[fancythm,fancybb]{jphmacros2e} 
\usepackage{amsfonts}
\usepackage{mathrsfs}
\usepackage{amsmath}
\usepackage{amsthm}
\usepackage{caption}
\usepackage{subfigure}
\usepackage{graphicx}

\newcommand{\blue}{\color{black}}

\newtheorem{lemma}{Lemma}
\newtheorem{theorem}{Theorem}

\begin{document}

	\title{Resilient control under denial-of-service and uncertainty: {\blue An adaptive dynamic programming} approach}
	\author{Weinan~Gao,~\IEEEmembership{Senior Member,~IEEE,}~Zhong-Ping~Jiang,~\IEEEmembership{Fellow,~IEEE},~Tianyou~Chai,~\IEEEmembership{Life Fellow,~IEEE}
		\thanks{W. Gao and T. Chai are with the State Key Laboratory of Synthetical Automation for Process Industries, Northeastern University, Shenyang, China, 110819. Email: \texttt{weinan.gao@nyu.edu, tychai@mail.neu.edu.cn}.}
		\thanks{Z. P. Jiang is with the Department of Electrical and Computer Engineering, Tandon School of Engineering, New York University, Brooklyn, NY 11201, USA. Email: \texttt{zjiang@nyu.edu}.}
		\thanks{The corresponding author is T. Chai.}
		\thanks{The work has been supported in part by the National Natural Science Foundation of China (NSFC) under Grant 62373090 and the U.S. National Science Foundation under Grant CNS-2227153.}
	}

	\markboth{}%
	{}
	%



	\maketitle
	
	\begin{abstract}
	In this paper, a new framework for the resilient control of continuous-time linear systems under denial-of-service (DoS) attacks and system uncertainty is presented. Integrating techniques from reinforcement learning and output regulation theory, it is shown that
	resilient optimal controllers can be learned directly from real-time state and input data collected from the systems subjected to attacks. Sufficient conditions are given under which the closed-loop system remains stable {\blue given any upper bound of DoS attack duration}.
	Simulation results are used to demonstrate the efficacy of the proposed learning-based framework for resilient control under DoS attacks and model uncertainty.
		
	\end{abstract}
	
	\begin{IEEEkeywords}
		Adaptive Dynamic Programming, Output Regulation, Denial-of-Service Attack, Resilient Optimal Control
	\end{IEEEkeywords}

	%
	\IEEEpeerreviewmaketitle

\section{Introduction}
As modern control engineering systems become more and more complex and uncertain, it is important to learn controllers from online data that avoids the reliance on the accurate model knowledge. 
A plausible strategy is to learn the dynamics of control systems first, and then design controllers based on the learned dynamics \cite{Ian2019TRO,yu2017preparing}. 
The success of this strategy usually requires the discrepancy between learned and actual dynamics to be small enough.  
The second strategy is to learn control policy directly from online data instead of learning dynamics \cite{ROTULO2022Auto,Jiang2020FnT,QASEM2023Auto,Jiang2022TNNLS}. Following this strategy, adaptive dynamic programming (ADP) has been used to search for the solution to an optimal decision making problem based on active agent-environment interactions \cite{HuaguangMagazine,Liu2021Survey,Vamvoudakis2020Survey,Powell2007,Lewis2018Survey,Gao2022CTT,Yang2021TSMC}.
Its efficiency has been proved especially in situations where the exact system model is difficult or even impossible to be acquired.
Apart from the unknown system dynamics, modern control systems with ADP algorithms must also deal with nonvanishing disturbance, dynamic uncertainty and cyber attacks.

By the integration of differential game theory and ADP, one can learn the Nash policy to attenuate the effect from nonvanishing disturbances \cite{vrabiegames2013optimal,ODEKUNLE2020}.
In order to further reject disturbance while asymptotically tracking some prescribed reference signal, ADP has been leveraged to achieve output regulation in an adaptive optimal sense \cite{Gao2016TAC,Chen2019,Jiang2020Cyber}.  
However, these methods usually consider static uncertainties in the system based on an assumption that the system order is known exactly. 
In order to relax this assumption, robust ADP (for short, RADP) has been proposed via a combination of ADP and nonlinear small-gain theory \cite{JiangBook2017}.  
Under some mild small-gain condition, one can learn a robust optimal control policy via RADP for unknown control systems in the presence of dynamic uncertainties. 
Nevertheless, most of existing (robust) ADP methods ignore the resiliency of control systems when the control system is subject to malicious attacks.

As a superset of ADP, reinforcement learning has been studied to improve the resilience of control systems while preserving the optimality or at least suboptimality; see \cite{Ilahi2022TAI,HUANG2022273} and references therein. 
However, most of these studies fail to guarantee the asymptotic stability of the closed-loop system. 
In our previous work, we have proposed a model-free method to analyze the resiliency and the stability of control systems in closed-loop with controllers learned by ADP \cite{GAO2022Auto}.
Without knowing exactly the system dynamics, 
one can provide an estimate on the upper bound of cyberattacks for which the control system remains operational. 
However, this method of \cite{GAO2022Auto} may fail to work whenever the cyberattacks go beyond our
estimation.
As cyberattacks become increasingly frequent in modern engineering systems, it is urgent to develop a novel learning-based controller design framework that
addresses resiliency besides the common issues of stability and optimality.

\subsection*{Contributions}
The contributions of this paper are summarized as follows. 
First, we develop an original learning-based resilient optimal controller design framework via ADP for a class of continuous-time linear systems affected by disturbances and denial-of-service (DoS) attacks. 
Interestingly, one can rely on this framework to design a resilient optimal control policy to achieve output regulation (i.e., asymptotic tracking with disturbance rejection) {\blue given any upper bound of DoS attack duration}.
Second, comparing with existing resilient control design approaches \cite{FENG2017Auto,DENG2022Auto,DePersis2015TAC,Tang2021TAC}, the overall control design process under the developed framework does not require the precise knowledge of system dynamics in the state equation. 
Third, this framework is a data-based generalization of the average dwell-time approach in model-based switched system theory
 \cite{Zhai2001,Liberzon2003}.

\subsection*{Structure}
The remainder of this paper is organized as follows. 
Section \ref{sec-prob} formulates an optimal output regulation problem with desired convergence rate and presents a model-based solution to this problem. 
In Section \ref{sec-online}, we propose a learning-based resilient optimal controller design framework and an algorithm based on the proposed framework. We have also rigorously analyzed the resiliency of the closed-loop system. 
Section \ref{sec-simu} contains simulation results on a benchmark adaptive cruise control problem of an autonomous vehicle under DoS attacks by leveraging the proposed algorithm.
Section \ref{sec-conclusion} closes the paper with brief concluding remarks.

\subsection*{Notations}
Throughout this paper,
$\mathbb{R}_+$ denotes the set of nonnegative real numbers, 
and $\mathbb{Z}_+$ the set of nonnegative integers.
$\mathbb{C}^-$ indicates the open left-half complex plane. 
$|\cdot|$ represents the Euclidean norm for vectors and the induced norm for matrices.
The symbol $\otimes$ indicates the Kronecker product operator and ${\rm vec}(A)=[a_1^T,a_2^T,\ldots,a_m^T]^T$, where $a_i\in\mathbb{R}^n$ are the columns of $A\in\mathbb{R}^{n\times m}$. 
When $n=m$, $\sigma(A)$ is the spectrum of $A$.
For an arbitrary column vector $v\in \mathbb{R}^n$,
${\rm vecv}(v)=[v_1^2,v_1 v_2,\ldots,v_1v_n,v_2^2,v_2v_3,\ldots,v_{n-1}v_n,v_n^2]^T\in\mathbb{R}^{\frac{1}{2}n(n+1)}.$
$ {\rm vecs}(P)=[p_{11},2p_{12},\ldots,2p_{1m},p_{22},2p_{23},\ldots,$
$2p_{m-1,m}$
$,p_{mm}]^T\in\mathbb{R}^{\frac{1}{2}m(m+1)}$ for a symmetric matrix $P\in\mathbb{R}^{m\times m}$, and $\lambda_{M}(P)$ (resp. $\lambda_{m}(P)$) denotes the maximum (resp. minimum) eigenvalue of a real symmetric matrix $P$.
$P\succ 0$ (resp. $P\prec 0$) represents that $P$ is positive (resp. negative) definite.

\section{Problem Formulation and Preliminaries} \label{sec-prob}

In this section, we will formulate a linear optimal output regulation problem and recall a model-based solution to this problem. Then, a DoS attack model is established and the concept of resilient optimal control is introduced under DoS attacks. 

\subsection{Linear Optimal Output Regulation}

Consider a class of continuous-time linear systems described by
\begin{align} 
	\dot x(t)      =   &   Ax(t)+Bu(t)+Dv(t), \label{system_state}      \\
	e(t)   =   &   Cx(t)+Fv(t) ,     \label{system_output}
\end{align}
where $u(t)\in \mathbb{R}$ and $x(t)\in \mathbb{R}^{n}$ are the control input and the state, respectively. 
The tracking error $e(t)\in\mathbb{R}$ is the difference between the output $Cx(t)$ and the reference $-Fv(t)$.
The signal $v(t)\in \mathbb{R}^q$ is the exostate of an autonomous system usually referred to as exosystem:
\begin{align}
	\dot v(t)=   &   S v(t).  \label{sys_exo}
\end{align}

The constant matrices $A$, $B$, $C$, $D$, $F$ and $S$ are with proper dimensions.
Throughout this paper, the following assumptions are made on the overall system (\ref{system_state})-(\ref{sys_exo}).

\begin{assumption} \label{assumAB}
	The pair $(A,B)$ is controllable.
\end{assumption}
\begin{assumption} \label{assumRank}
	{\rm rank}$\left[
	\begin{array}{cc}
		A-\lambda I & B \\
		C & 0 \\
	\end{array}
	\right]=n+1$, $\forall\lambda\in\sigma(S)$.
\end{assumption}

The output regulation is achieved if the system (\ref{system_state})-(\ref{sys_exo}) in closed-loop with a controller is exponentially stable at the origin when $v\equiv 0$, and the tracking error $e(t)$ asymptotically converges to $0$ for all initial conditions $x(0)\in\mathbb{R}^n$ and $v(0)\in\mathbb{R}^q$.
Assume that there is no DoS attacks, the closed-loop system with the following controller
\begin{align} \label{feedback-con}
	u(t) = -K x(t) + L v(t), 
\end{align}
has achieved output regulation, where $K$ is a feedback control gain such that ${\blue \sigma(A-BK)\in\mathbb{C}^-,  L=KX+U}$ is a feedforward control gain
{\blue with the matrix $X\in\mathbb{R}^{n\times q}$ and the vector $U\in\mathbb{R}^{1\times q}$ solving the following regulator equations} 
\begin{align}
	XS=& AX+BU+D, \label{eqn-regu1}\\
	0 =&  C X + F. \label{eqn-regu2}
\end{align}

Based on output regulation theory \cite{huangjiebook}, $Xv(t)$ and $Uv(t)$ are where the state and the input should fall in steady-state if one hopes to solve the linear output regulation problem. 
For any $t\in\mathbb{R}_+$,	
$\tilde x(t)= x(t)-Xv(t)$ and $\tilde u(t)=u(t)-U v(t)$ are differences between the actual state/input and the corresponding steady-state components.
Since ${\blue \sigma(A-BK)\in\mathbb{C}^-}$, one can show that $\underset{t\to\infty}{\lim}\tilde x(t)=0$, $\underset{t\to\infty}{\lim}\tilde u(t)=0$, and $\underset{t\to\infty}{\lim}e(t)=0$, which implies that the systems (\ref{system_state})-(\ref{feedback-con}) achieve output regulation.

{\blue
\begin{remark}
	 Under the conditions in the Assumption \ref{assumRank}, there always exists uniquely a pair $(X,U)$ solving regulator equations (\ref{eqn-regu1})-(\ref{eqn-regu2}); see \cite{huangjiebook}. 
\end{remark}
}

In order to take into account the transient performance and the convergence rate of the closed-loop system, an optimal output regulation problem is stated as follows. 
		\begin{problem} \label{prob-LQ} 
			The optimal output regulation problem is solved if one can find a feedback control policy $u^*$ to achieve output regulation. Moreover, the following dynamic programming
				\begin{align}
						&\underset{\tilde u}{\min}\int_{0}^{\infty} e^{2t\lambda^- }( \tilde x^T(t) Q \tilde x(t)+\tilde u^2(t))dt \label{eqn-cost}\\
						&{\rm s.t.\quad}\dot {\tilde x}(t)= A \tilde x(t) + B \tilde u(t)  \label{eqn-open}
					\end{align}
					is solved by $\tilde u^*=u^*-Uv$, where $Q\succ 0$ and $\lambda^-\in\mathbb{R}_+$.
				\end{problem}
				
				By denoting $\tilde x_{\lambda^-}(t)=e^{t\lambda^-}\tilde x(t)$ and 
				$\tilde u_{\lambda^-}(t)=e^{t\lambda^-}\tilde u(t)$, we have 
				\begin{align} \label{eqn-rate}
					\dot {\tilde x}_{\lambda^-}(t) = & e^{t\lambda^-}\dot {\tilde x}(t) + \lambda^- \tilde x_{\lambda^-}(t) \nonumber\\
					=&  e^{t\lambda^-}(A\tilde x(t) + B \tilde u(t)) + \lambda^- \tilde x_{\lambda^-}(t) \nonumber\\
					=&(A+I\lambda^- ) \tilde x_{\lambda^-}(t) + B \tilde u_{\lambda^-}(t)
				\end{align}
				
				By (\ref{eqn-rate}) and the fact that 
				the cost in (\ref{eqn-cost}) is equivalent to 
				\begin{align}
					\int_{0}^{\infty} ( \tilde x_{\lambda^-}^T(t) Q \tilde x_{\lambda^-}(t)+\tilde u_{\lambda^-}^2(t))dt, \nonumber
				\end{align}
			    one can observe that Problem \ref{prob-LQ} is convertible to a linear quadratic regulator problem. The solution to this problem is an optimal feedback control policy
				\begin{align} \label{eqn-optimal2}
					\tilde u^*_{\lambda^-}(t)=-K^* \tilde x_{\lambda^-}(t),
				\end{align}
				which is equivalent to 
				\begin{align} \label{eqn-optimal3}
					u^*(t)=&e^{-t\lambda^-}\tilde u^*_{\lambda^-}(t)+U v(t) \nonumber\\
					=& -K^* e^{-t\lambda^-} \tilde x_{\lambda^-}(t)+U v(t)  \nonumber\\
					=& -K^* \tilde x(t)+U v(t) \nonumber\\
					=& -K^* x(t)+L^* v(t),
				\end{align}
				where the optimal feedforward control gain $L^*$ is 
				\begin{align}
					L^* = K^* X + U.
				\end{align}
				
				The optimal control gain $K^*$ is
				\begin{align}
					K^*= B^T P^*.
				\end{align}
				
				{\blue Based on the Popov-Belevitch-Hautus (PBH) test, the Assumption 1, and $Q\succ 0$, one can obtain that the pair $(A-I\lambda^-,B)$ is controllable and the pair $(A-I\lambda^-,\sqrt{Q})$ is observable for any $\lambda^-\in\mathbb{R}_+$.
				From linear optimal control theory, the real symmetric matrix $P^*\succ 0$ is the unique solution to the following algebraic Riccati equation (ARE)}
				\begin{align} \label{eqn-ARE2}
					(A+I\lambda^- ) P^*+P^* (A+I\lambda^- )+Q-P^* BB^T  P^*=0.
				\end{align}

\subsection{DoS Attack Model and Resilient Optimal Control}
The DoS attack model considered in this paper can be described as follows. Let $\mathcal{I}_{l}=[h_l,h_l+\tau_l)$ represent the $l$-th  ($l\in\mathbb{Z}_+$) DoS attack interval.
$h_l$, $h_l+\tau_l$ and $\tau_l$ represent the start time, end time and the length of the $l$-th DoS attack.
Based on the following Assumptions \ref{assum_A_DoS_F}-\ref{assum_A_DoS_1} that are somehow introduced in \cite{DePersis2015TAC,Tang2021TAC,FENG2017Auto,DENG2022Auto,WANG2022Auto}, we define $\Pi_D(t_a,t_b):=\ (t_a,t_b) \bigcap \bigcup\limits_{l=0}^{\infty}\mathcal{I}_{l}$ as the set where the system is under DoS attacks during the interval $[t_a,t_b]$, and denote $\Pi_N(t_a,t_b):=[t_a,t_b]\setminus\Pi_D(t_a,t_b)$ as the normal communication set.

\begin{assumption} (DoS Frequency) \label{assum_A_DoS_F} 
	There exist constants $\eta>1$ and $\tau_D>0$ such that 
	\begin{align}
		n(t_a,t_b)\leq \eta+\frac{t_b-t_a}{\tau_D},~\forall~ t_b> t_a\geq0,
	\end{align}
	where $n(t_a,t_b)$ denotes the number of DoS off/on transitions during the interval $[t_a,t_b]$.
\end{assumption}

\begin{assumption}\label{assum_A_DoS_1} 
	(DoS Duration)
	There exist constants $T>1$ and $\kappa>0$ such that
	\begin{align*}
		\left\vert\Pi_D(t_a, t_b)\right\vert \leq \kappa+\frac{t_b-t_a}{T} ,~\forall~ t_b> t_a\geq0,
	\end{align*}
	where $\vert\Pi_D(t_a, t_b)\vert$ is the Lebesgue measure \cite{Rudin1953} of the set $\Pi_D(t_a, t_b)$.
\end{assumption}



Similar to robust optimal with respect to systems with dynamic uncertainties \cite{JiangBook2017}, we define the resilient optimality of a control policy with respect to systems under DoS attacks.
\begin{definition}
A control policy $u_{ro}^*$ is said to be resilient optimal regarding to systems (\ref{system_state})-(\ref{sys_exo}) 
if both following properties hold.
\begin{enumerate}
	\item The optimal output regulation problem, Problem \ref{prob-LQ}, is solved by $u_{ro}^*$ in the absence of DoS attacks, i.e., $|\Pi_D(0,t_b)|=0$ for any $t_b\in\mathbb{R}_+$.
	\item The closed-loop systems (\ref{system_state})-(\ref{sys_exo}) with $u_{ro}^*$ achieve output regulation in the presence of DoS attacks satisfying Assumptions \ref{assum_A_DoS_F}-\ref{assum_A_DoS_1}.
\end{enumerate}
\end{definition}

\section{Main Results} \label{sec-online}
					
		In this section, we will establish a learning-based resilient optimal control framework for control systems subject to DoS attacks.
		Based on this framework, we develop a learning algorithm to find a resilient optimal control policy {\blue given any upper bound of DoS attack duration}. We further show that the closed-loop system achieves output regulation. 
					
					\subsection{Learning-based Controller Design Under DoS}  \label{sec-single}
		 A system with DoS attacks is essentially an autonomously switched system where we have no direct control over the switching mechanism \cite{Liberzon2003}. We will switch between two linear subsystems, one is controllable (when the communication is normal) while the other one is an autonomous system but not necessarily stable (when the communication is denied).
					As shown in \cite{Zhai2001}, a sufficient condition to ensure the exponential stability of this kind of switched systems is that the total activation time ratio between an exponentially stable linear subsystem and an unstable linear subsystem is smaller than a criterion related to the eigenvalues of each subsystem and that the average dwell time is sufficiently large. 
															
					\begin{figure}[t]
						\centering
						\includegraphics[width=\linewidth]{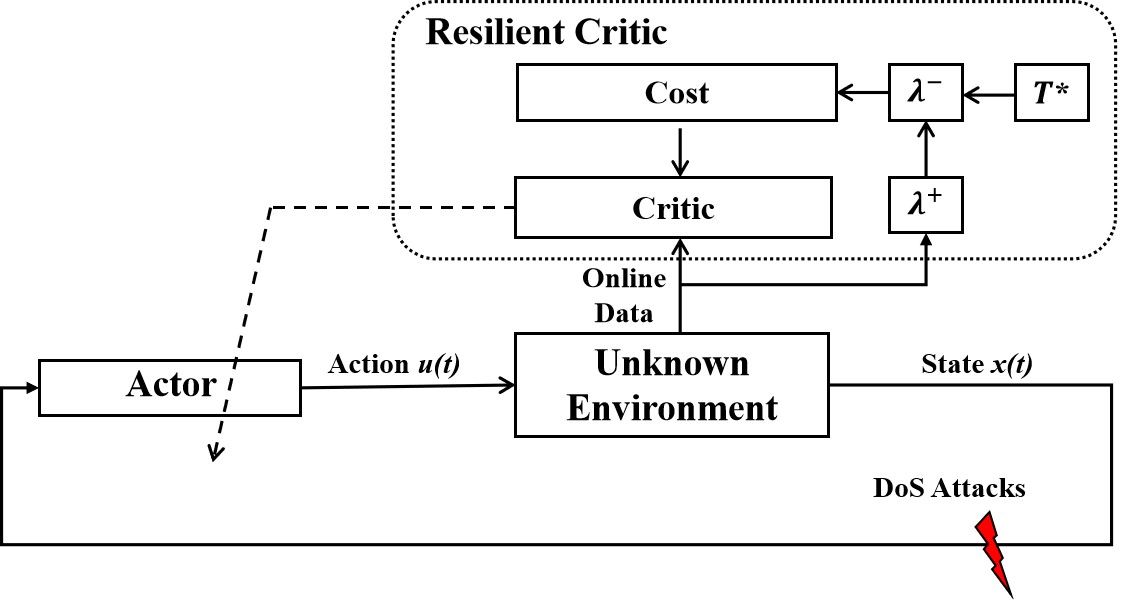}
						\caption{Learning-based resilient optimal control framework}
						\label{fig:framework}
					\end{figure} 
					
					Based on above-mentioned fact, we develop a learning-based resilient optimal control framework, see Fig. \ref{fig:framework}.
					The outline of this framework can be understood as follows.
					Through state and input data, we first learn a $\lambda^+\in\mathbb{R}_+$ such that {\blue $\sigma (A-\lambda^+I)\in\mathbb{C}^-$}.
					As the system is open-loop when it is under DoS attacks, $\lambda^+$ can be used to quantify the divergence rate of the system under attack.
					Then, we determine the convergence rate parameter $\lambda^-$ for the system with normal communication {\blue based on DoS duration criterion $T$}, and choose the cost (\ref{eqn-cost}) based on $\lambda^-$.
					Finally, we develop a resilient optimal controller based on actor-critic structure and policy iteration to satisfy this sufficient condition without the exact knowledge of system matrices $A, B$ and $D$.
					
					To implement this framework, we begin with solving the following Lyapunov equation  
					\begin{align} \label{lyap-open}
						(A-\lambda^+I)^T P^+ + P^+ (A-\lambda^+I) = -\epsilon I
					\end{align}
					using input and state data, where $\epsilon$ is a small but positive constant. Notice that one can always solve the (unique) real symmetric matrix $P^+\succ 0$ from (\ref{lyap-open}) once {\blue $\sigma (A-\lambda^+I)\in\mathbb{C}^-$}.
					Since the state matrix $A$ is unknown, one cannot solve (\ref{lyap-open}) directly. 
					We propose a way to calculate $P^+$ through input and state data.
					From (\ref{system_state}), for any symmetric matrix $P\in\mathbb{R}^{n\times n}$, we have
					\begin{align} \label{eqn-xPx}
						& x^T(t_1) P x(t_1)-  x^T(t_0)P x(t_0) \nonumber\\
						=&\int_{t_0}^{t_1}\left[x^T(A^TP+PA)x+2x^T P B u  + 2 x^T P D v\right.] d\tau.
					\end{align}
				
					Suppose there exists a sequence $\{t_i\}_{i=0}^\infty$ such that the communications are allowed in all of the following intervals $[t_0,t_1], [t_2,t_3], [t_4,t_5], \cdots$. 
					It is feasible to find this sequence through detecting attacks; see, e.g., \cite{Zhou2021CYB}. 
					For any two vectors $a,b$ and a sufficiently large number $s>0$, define
					\begin{align} 
						\delta_{a}=& [{\rm vecv}(a(t_{1}))-{\rm vecv}(a(t_{0})),{\rm vecv}(a(t_{3}))
						-{\rm vecv}(a(t_{2})), \nonumber\\
						& \ldots,{\rm vecv}(a(t_{2s+1}))-{\rm vecv}(a(t_{2s}))]^T, \nonumber\\
						\Gamma_{a,b}=& [\int_{t_{0}}^{t_{1}}a\otimes b d\tau,\int_{t_{2}}^{t_{3}} a\otimes b d\tau,\ldots ,\int_{t_{2s}}^{t_{2s+1}} a\otimes b d\tau]^T. \nonumber\\
						\Lambda_{a} = & [\int_{t_{0}}^{t_{1}}{\rm vecv}(a) d\tau,\int_{t_{2}}^{t_{3}} {\rm vecv}(a) d\tau,\ldots ,\int_{t_{2s}}^{t_{2s+1}} {\rm vecv}(a) d\tau]^T.
						\nonumber
					\end{align}
					
					The equations (\ref{lyap-open}) and (\ref{eqn-xPx}) imply the following linear equation
					\begin{equation} \label{myPolicy Evaluation3}
						\Psi \left[
						\begin{array}{c}
							{\rm vecs}(P^+) \\
							{\rm vec}(B^T P^+) \\
							{\rm vec}(D^T P^+) \\
						\end{array}
						\right] = \Phi,
					\end{equation}
					where
					$\Psi=[\delta_{x}-2\lambda^+\Lambda_{x},-2{\Gamma_{x,u}},-2{\Gamma_{x,v}}], 
					\Phi=-\epsilon{\Gamma_{x,x}}.$

					Then, we choose the desired convergence rate parameter $\lambda^-$ of the closed-loop system with normal communication and obtain the quadruple $(K^*, P^*,X,U)$ based on online data. 
					In order to find the pair $(X,U)$ which is a solution to the regulator equations (\ref{eqn-regu1})-(\ref{eqn-regu2}), we define a sequence of trial matrices $X_i\in\mathbb{R}^{n\times q}$ for $i=0,1,\cdots, (n-r)q+1$, where $X_0=0$, $CX_1=-F$. For $i=2,3,\cdots, (n-r)q+1$, all the vectors ${\rm vec}(X_i)$ form a basis for the kernel $\ker(I_q\otimes C)$. It has been shown in \cite{Gao2016TAC} that the solution $X$ is always a linear combination of these trial matrices.
					

					The idea of policy iteration is to implement both policy evaluation 
					\begin{align}
						0=&A_k^T P_k+P_kA_k+Q +K_k^TK_k+2\lambda^- P_k \label{lyap} 		
					\end{align}
					and policy improvement 
					\begin{align}
						K_{k+1}=& B^T P_k \label{Klein-gain}
					\end{align}
					using online data by iterations, where $A_k=A-BK_k$.
					
					Defining $x_{i}(t)=x(t)-X_i v(t)$, along the trajectories of (\ref{system_state})-(\ref{sys_exo}), we have
					\begin{align}
						\dot x_{i}=& \dot x - X_i Sv \nonumber\\
						=&Ax_{i}+Bu+(D-\mathcal{S}(X_i))v \nonumber\\
						=&{\blue A_k x_i+B(u+K_kx_i)+(D-\mathcal{S}(X_i))v }
						\label{closed_+1}
					\end{align}
					where $\mathcal{S}(X_i)=X_iS-AX_i$.
					
					By equations (\ref{eqn-xPx}), (\ref{lyap}) and (\ref{Klein-gain}), we have
					\begin{align} \label{eqn-adp}
						& x_i^T(t_1)P_k x(t_1)-  x_i^T(t_0)P_k x(t_0) \nonumber\\
						=&\int_{t_0}^{t_1}\left[x_i^T(A_k^TP_k+P_kA_k)x_i+2x_i^T P_k B {\blue(u+K_kx_i)} \right. \nonumber\\
						& \left. + 2 x_i^T P_k (D-\mathcal{S}(X_i))v\right.] d\tau \nonumber\\ 
						=& \int_{t_0}^{t_1}\left[-x_i^T(Q+K_k^T RK_k+2\lambda^-P_k)x_i \right. \nonumber\\
						& \left. {\blue +2x_i^T K_{k+1}^T (u+K_kx_i)} + 2 x_i^T P_k (D-\mathcal{S}(X_i))v\right.] d\tau.  
					\end{align}

					The equation (\ref{eqn-adp}) implies the following linear equation
					\begin{equation} \label{myPolicy Evaluation2}
						\Psi_{ik} \left[
						\begin{array}{c}
							{\rm vecs}(P_k) \\
							{\rm vec}(K_{k+1}) \\
							{\rm vec}\left( (D-\mathcal{S}(X_i))^T P_k\right) \\
						\end{array}
						\right] = \Phi_{ik},
					\end{equation}
					where
					\begin{align}
						\Psi_{ik}= & [\delta_{x_i}+2\lambda^-\Lambda_{x_i},-2{\Gamma_{x_i,x_i}}(I\otimes K_k^T)-2{\Gamma_{x_i, u}},-2{\Gamma_{x_i,v}}], \nonumber\\ 
						\Phi_{ik}=& -{\Gamma_{x_i,x_i}} {\rm vec}(Q+K_k^TK_k). \nonumber
					\end{align}
					
					A sufficient condition to ensure equations (\ref{myPolicy Evaluation3}) and (\ref{myPolicy Evaluation2}) having a unique solution is that matrices $\Psi$ and $\Psi_{i,k}$ are full-column rank. It is obvious that matrix $\Psi_{i,k}$ depends on the control gain $K_k$.
					In the following lemma, we will show that the uniqueness of solution to equations (\ref{myPolicy Evaluation3}) and (\ref{myPolicy Evaluation2}) can be guaranteed based on some rank condition that purely depends on the online collection. 
					
					\begin{lemma} \label{lemma-rank}
						For $i=0,1,2,\cdots,(n-r)q+1$, if there exists a $s^*\in\mathbb{Z}_+$ such that for all $s>s^*$,
					\begin{align} \label{fullrank}
						{\rm rank}([\Gamma_{x_i, x_i}, \Gamma_{x_i, u}, \Gamma_{x_i, v}])=\frac{(n+1)n}{2}+(m+q)n.
					\end{align}
						then, the following properties hold
						\begin{enumerate}
							\item $\Psi$ has full column rank for any $\lambda^+$ such that {\blue $\sigma (A-\lambda^+I)\in\mathbb{C}^-$}.
							\item $\Psi_{ik}$ has column full rank for any $k\in\mathbb{Z}_+$.
						\end{enumerate}
					\end{lemma}
		 
					%
					
					\begin{proof}
						To show the property 1), let us assume there exists a vector $\Xi$ satisfying the following equation 
						\begin{align} \label{eqn-PsiXi}
							\Psi \Xi = \Phi 
						\end{align}
						where $\Xi = [({\rm vecs}(W))^T, ({\rm vec}(Y))^T, ({\rm vec}(Z))^T]^T$ with matrices $W=W^T\in\mathbb{R}^{n\times n}$, $Y\in\mathbb{R}^{m\times n}$ and $Z\in\mathbb{R}^{q\times n}$.
						The equation (\ref{eqn-PsiXi}) can be converted as follows based on (\ref{eqn-xPx})
						\begin{align}
							\Lambda_{x} {\rm vecs}(\Omega_1) + 2 \Gamma_{x,u} {\rm vec}(\Omega_2) + 2 \Gamma_{x,v} {\rm vec} (\Omega_3) = 0
						\end{align}
						where 
						\begin{align}
							\Omega_1 = & (A-\lambda^+I)^T W+W(A-\lambda^+I)+\epsilon I, \nonumber\\
							\Omega_2 = & B^T W - Y, \nonumber\\
							\Omega_3 = & D^T W - Z. \nonumber
						\end{align}
					
						When $i=0$, the rank condition (\ref{fullrank}) implies that 
							the matrix $[\Lambda_{x}, \Gamma_{x, u}, \Gamma_{x, v}]$ is full column rank,
							whereby one can show that $Y=B^TW$, $Z=D^TW$ and $W$ solves the following Lyapunov equation 
						\begin{align} \label{eqn-lyapW}
							(A-\lambda^+I)^T W+W(A-\lambda^+I)=-\epsilon I.
						\end{align}
					
						Since {\blue $\sigma (A-\lambda^+I)\in\mathbb{C}^-$}, the matrix $W$ is uniquely determined. Thus, the vector $\Xi$ is the unique solution to (\ref{eqn-PsiXi}) and $\Psi$ is full column rank. The proof of property 1) is thus completed. To show property 2), one can refer to \cite{Gao2016TAC} for mathematical proof.
					\end{proof}
				
					Based on Lemma \ref{lemma-rank}, we further show a necessary and sufficient condition related to the stability of matrix $A-\lambda^+I$ in the following lemma so that its stability can be determined without the precise knowledge of system matrices in the Algorithm \ref{algor-PI}, an off-policy learning algorithm to be developed below.
					
					\begin{lemma} \label{lemma-eqv}
						Under the rank condition (\ref{fullrank}), the matrix $A-\lambda^+I$ is Hurwitz if and only if the following properties hold
						\begin{enumerate}
							\item The matrix $\Psi$ is full column rank.
							\item The symmetric matrix $P^+$ solved by (\ref{myPolicy Evaluation3}) is positive definite.
						\end{enumerate}
					\end{lemma}
				
					\begin{proof}
						To prove the sufficiency, suppose $A-\lambda^+I$ is Hurwitz, then the property 1) holds based on Lemma \ref{lemma-rank} and rank condition (\ref{fullrank}). We further observe that $P^+$ can be solved uniquely by (\ref{myPolicy Evaluation3}) when $\Psi$ is full column rank. On the other hand, the matrix $P^+\succ 0$ solving the Lyapunov equation (\ref{lyap-open}) satisfies  (\ref{myPolicy Evaluation3}), which immediately shows the property 2).
						
						To prove the necessity, suppose both properties hold. From equations (\ref{eqn-PsiXi})-(\ref{eqn-lyapW}), the matrix $P^+\succ 0$ solves the Lyapunov equation (\ref{lyap-open}). By Lyapunov stability theory, $A-\lambda^+I$ must be a Hurwitz matrix. The proof is thus completed. 
					\end{proof}
					
					Now, we are ready to present an {\blue off-policy} resilient PI Algorithm \ref{algor-PI}, which essentially works through two phases. 
					In the first phase (Steps 3-7), we can use Lemma \ref{lemma-eqv} to ensure that one can always find a $\lambda^+$ such that $A-\lambda^+I$ is Hurwitz via online collection.
					Based on $\lambda^+$ and DoS duration criterion $T$, we set a convergence rate parameter $ \lambda^-=\frac{2\lambda^+}{T-1}$ for the closed-loop system with normal communication, and
					approximate the resilient optimal control gain $K^*$ and the value $P^*$, and the solution to regulator equations in the second phase (step {\blue 8-19}). 
					
					\begin{algorithm} 
						\caption{{\blue Off-Policy} Resilient Policy Iteration Algorithm} \label{algor-PI}
						\begin{algorithmic} [1]
							\State {\blue Select a small $c>0$, and two constants $\lambda^+>0,a>1$}. Apply any locally essentially bounded input $u$ on $[t_0,t_{2s+1}]$ such that (\ref{fullrank}) holds. 
							\State $i\gets 0$, $P^+\gets 0$.
							\Repeat	
							\State ${\blue \lambda^+\gets a\lambda^+}$	
							\State \textbf{if} $\Psi$ is full column rank \textbf{then}
							\State Solve $P^+$ from (\ref{myPolicy Evaluation3}) \textbf{end if}
							\Until $P^+\succ 0$  
							\State  ${\blue k\gets -1, P_k\gets 0}$, $\lambda^-\gets\frac{2\lambda^+}{T-1}.$ {\blue Choose an admissible $K_0$.}
							\Repeat		
							\State $k\gets k+1$
							\State Solve $(P_k,K_{k+1})$ from (\ref{myPolicy Evaluation2})
							\Until  $|P_k-P_{k-1}|<c$ \textbf{\blue and }${\blue k>0}$
							\State ${\blue k^*\gets k}$
							\Repeat 
							\State Solve $\mathcal{S}(X_i)$ from (\ref{myPolicy Evaluation2})
							\State $i\gets i+1$
							\Until $i=(n-r)q+1$
							\State Solve $(X,U)$ from $\mathcal{S}(X_i)$. ${\blue L_{k^*}=K_{k^*}X+U}$ 
							\State Return the resilient near-optimal control policy
							\begin{align} \label{eqn-learn-con}
								{\blue u_{k^*}(t)= \left\{  \begin{matrix}
									-K_{k^*}x(t)+L_{k^*}v(t), & t\notin \mathcal{I}_s \\ Uv(t), & t\in \mathcal{I}_s
								\end{matrix}   \right.} 
							\end{align}
						\end{algorithmic}
					\end{algorithm}

					\begin{remark}
						Given $S(X_i)$ for $i=0,1,2,\cdots,(n-r)q+1$, one can solve $(X,U)$ directly by tackling a convex optimization problem without relying on the knowledge of system matrices; one can refer to \cite[Remark 5]{Gao2016TAC} for more details.
					\end{remark}
					
					{\blue \begin{remark}
						A control gain $K_0$ is said to be admissible if it 1) is stabilizing and 2) makes the cost (\ref{eqn-cost}) of the system (\ref{eqn-open}) with the controller $\tilde u(t)=-K_0\tilde\xi(t)$ to be finite. Since the system (\ref{eqn-open}) is linear and the cost (\ref{eqn-cost}) is quadratic, one can show that any $K_0$ satisfying $\sigma(A+I\lambda^- - BK_0)\in\mathbb{C}^-$ is admissible. If an admissible $K_0$ is not available, one can use the idea of hybrid iteration \cite{QASEM2023Auto}, which is a combination of value iteration and PI, for learning.  
					\end{remark}
				}

								{\blue We will develop the following theorem to show that the system in closed-loop with 
									\begin{align} \label{eqn-control-uk}
									u_{k}(t)= \left\{  \begin{matrix}
											-K_{k}x(t)+L_{k}v(t), & t\notin \mathcal{I}_s \\ Uv(t), & t\in \mathcal{I}_s
										\end{matrix}   \right. 
									\end{align}
								}is always resilient to DoS attacks if DoS frequency criterion $\tau_D$ is sufficiently large, in other words, if DoS frequency is small enough, {\blue where the feedforward control gain is $L_k=K_kX+U$. The feedback control gain $K_k$ is obtained at the iteration $k>0$ and the pair $(X,U)$ is obtained at the step 18 of the Algorithm \ref{algor-PI}. Without loss of generality, the resilient near-optimal control policy (\ref{eqn-learn-con}) can be rewritten in the form of (\ref{eqn-control-uk}) by letting $k=k^*$. }
								\begin{theorem} \label{theo-res}
								Under Assumptions \ref{assumAB}-\ref{assum_A_DoS_1}, {\blue given any DoS duration criterion $T$},
								 the system (\ref{system_state})-(\ref{sys_exo}) in closed-loop with the learned controller {\blue (\ref{eqn-control-uk})} achieves output regulation if the DoS frequency criterion $\tau_D$ satisfies
									\begin{align} \label{eqn-DoS-T}
										\tau_D \geq \log\left({\blue\frac{\lambda_{M}(P_{k})\lambda_{M}(P^+)}{\lambda_{m}(P_{k})\lambda_{m}(P^+)}}\right){\blue\frac{T}{ 2(T-1)\lambda^-}}:= {\blue \tau_D^k}.
									\end{align}
								\end{theorem} 
								
								\begin{proof}
									The system (\ref{system_state})-(\ref{sys_exo}) in closed-loop with the controller (\ref{eqn-learn-con}) can be rewritten by
									\begin{align} \label{eqn-closedloop}
										\dot {\tilde x}(t)=\left\{
										\begin{matrix}
											{\blue A_{k}} \tilde x(t), & t\notin \mathcal{I}_s  \\
											A \tilde x(t),        & t\in \mathcal{I}_s
										\end{matrix}
										\right.
									\end{align}
									
									Define the following piecewise-quadratic Lyapunov candidate 
									\begin{align}
										V(t)=\left\{
										\begin{matrix}
											\tilde x^T(t) {\blue P_{k}} \tilde x(t):=V_1, & t\notin \mathcal{I}_s \\
											\mu_1\tilde x^T(t) P^+ \tilde x(t):=V_0,        & t\in \mathcal{I}_s
										\end{matrix}
										\right. 
									\end{align}
									where $\mu_1>0$ is to be determined.
									
										%
											{\blue When the rank condition (\ref{fullrank}) holds, it can be observed from the Lemma \ref{lemma-rank} that the matrix $P_{k}=P_k^T\succ 0$ can be uniquely solved from (\ref{myPolicy Evaluation2}) and it satisfies (\ref{lyap}).
										}
											When the communication is normal, along the trajectories of the closed-loop system (\ref{eqn-closedloop}), we have
											\begin{align} 
												\dot V_1 = &  -\tilde x^T {\blue(A_{k}^TP_{k}+P_{k}A_{k})}\tilde x \nonumber\\
												\leq & -\tilde x^T (Q+{\blue K_{k}^T K_{k}}) \tilde x-2\lambda^-\tilde x^T{\blue P_{k}}\tilde x \nonumber\\
												\leq & -\left(2\lambda^-+c_1\right) V_1, t\notin \mathcal{I}_s.  \label{eqn-derive1}
											\end{align}
											where $c_1={\lambda_m(Q)/\lambda_M({\blue P_{k}})}$.
											
											On the other hand, when the system is under DoS attacks, based on (\ref{lyap-open}),  we have 
											\begin{align}
												\dot V_0 = & \mu_1\tilde x^T (A^TP^+ + P^+ A) \tilde x \leq  2\lambda^+ V_0, t\in \mathcal{I}_s.
											\end{align}
											
											For any $l\in\mathbb{Z}_+$, consider two intervals: $[h_l+\tau_l,h_{l+1})$ for the situation where the communications are normal and $[h_l,h_l+\tau_l)$
											for the situation where the communication is denied. Then, one can observe that the Lyapunov function satisfies
											\begin{align}
												V(t)\leq \left\{
												\begin{matrix}
													e^{-(2\lambda^-+c_1) (t-h_l-\tau_l)}V(h_l+\tau_l), \ t\in[h_l+\tau_l,h_{l+1})  \\
													e^{2\lambda^+ (t-h_l)}V(h_l),         \ t\in [h_l,h_l+\tau_l)
												\end{matrix}
												\right. \nonumber
											\end{align}
											
											Based on the definition of Lyapunov functions $V_0$ and $V_1$, we have
											\begin{align}
												& \mu_1 \lambda_m(P^+) |\tilde x|^2 \leq V_0(\tilde x) \leq \mu_1 \lambda_M(P^+) |\tilde x|^2, \nonumber\\
												& {\blue \lambda_m(P_{k}) |\tilde x|^2 \leq V_1(\tilde x) \leq \lambda_M(P_{k}) |\tilde x|^2}, \quad \forall \tilde x\in\mathbb{R}^n. \nonumber
											\end{align}
											
											This implies that 
											\begin{align}
												V_1(\tilde x) & \leq \mu_2 V_0(\tilde x), \nonumber\\
												V_0(\tilde x) & \leq \mu_2 V_1(\tilde x), \quad \forall \tilde x\in\mathbb{R}^n \nonumber
											\end{align}
											where $\mu_1={\blue \sqrt{\frac{\lambda_{M}(P_{k})\lambda_{m}(P_{k})}{\lambda_{M}(P^+)\lambda_{m}(P^+)}}}$, and 
											$\mu_2={\blue \sqrt{\frac{\lambda_{M}(P_{k})\lambda_{M}(P^+)}{\lambda_{m}(P_{k})\lambda_{m}(P^+)}}}$.
											
											Therefore, $V(t)\leq \mu_2 V(t^-)$ holds on all the switching point $t$, where $t^-=\lim_{\tau\uparrow t}\tau$.
											For all $t\in\mathbb{R}_+$, the Lyapunov function $V$ satisfies
											\begin{align}
												V(t)\leq & \mu_2^{n(0,t)} e^{-(2\lambda^-+c_1)|\Pi_N(0,t)|}e^{2\lambda^+|\Pi_D(0,t)|} V(0)\nonumber\\
												= & e^{-2(\lambda^-+c_1)|\Pi_N(0,t)|+2\lambda^+|\Pi_D(0,t)|+n(0,t) \log \mu_2  } V(0). \nonumber
											\end{align}
											
											Based on the Assumptions \ref{assum_A_DoS_F}-\ref{assum_A_DoS_1}, we have 
											\begin{align}
												|n(0,t)| \leq & \eta + \frac{t}{\tau_D}, \nonumber\\
												|\Pi_D(0,t)|\leq & \kappa + \frac{t}{T} {\leq  \kappa + \frac{t}{T}},  \nonumber\\
												|\Pi_N(0,t)|= & t-|\Pi_D(0,t)| \geq t-\kappa-\frac{t}{T},  \forall t\in\mathbb{R}_+.
											\end{align}
											
											Based on the fact that $\lambda^-=\frac{2\lambda^+}{T-1}$, one has
											\begin{align}
												&-2(\lambda^-+c_1)|\Pi_N(0,t)|+2\lambda^+|\Pi_D(0,t)|+n(0,t) \log \mu_2  \nonumber\\
												\leq & -2(\lambda^-+c_1) (t-\kappa-\frac{t}{T}) \nonumber\\
												& +2\lambda^+(\kappa + \frac{t}{T})+ (\eta+\frac{t}{\tau_D})\log \mu_2 \nonumber\\
												= & c_2 + t\left(-\frac{(\lambda^-+c_1)({T}-1)}{T}+\frac{\log \mu_2}{\tau_D}\right) \nonumber\\
												\leq &  c_2 - c_3 t,\quad \forall \tau_D\geq {\blue \tau_D^k}, t\in\mathbb{R}_+
											\end{align}
											where $c_2=2\kappa (\lambda^+ + \lambda^- + c_1)+\eta\log\mu_2$, and $c_3=\frac{c_1({T}-1)}{T}$.
											
											Finally, along the trajectories of the closed-loop system (\ref{eqn-closedloop}), the Lyapunov function $V$ satisfies 
											\begin{align}
												V(t)\leq & e^{c_2 - c_3 t}V(0), \quad \forall \tau_D\geq {\blue \tau_D^k} \nonumber
											\end{align}
											implying that 
											\begin{align}  \label{eqn-xi-ISS}
												|\tilde x(t)|\leq &c_4 e^{-\frac{c_3}{2} t}|\tilde x(0)|, \nonumber\\
												{ |e(t)|} \leq & |C|c_4e^{-\frac{c_3}{2} t}|\tilde x(0)|, \quad\forall \tau_D\geq {\blue \tau_D^k}
											\end{align}
											where 
											\begin{align}
												c_4={\blue \sqrt{e^{c_2}\frac{\max\{\lambda_M(P_{k}), \mu_1\lambda_M(P^+)\}}
													{\min\{\lambda_m(P_{k}), \mu_1\lambda_m(P^+)\}}} }. \nonumber
											\end{align}
											
											We can further have $\underset{t\to\infty}{\lim}\big(x(t)-Xv(t)\big)=0$, and $\underset{t\to\infty}{\lim}e(t)=0$,
											which shows that the output regulation is achieved. The proof is completed. 
										\end{proof}
									
										\begin{remark}
											In order to implement the learned resilient optimal control policy derived from the Algorithm \ref{algor-PI}, we assume the exostate $v(t)$ is always measurable.
											If $v(t)$ is unmeasurable, one can reconstruct the exostate by the minimal polynomial of exosystem matrix $S$; see \cite{Gao2016TAC} for more details. 
										\end{remark}

										\begin{remark}
											From (\ref{eqn-DoS-T}) in the Theorem \ref{theo-res}, we observe that the lower bound of DoS frequency criterion ${\blue \tau_D^k}$ is only related to ${\blue P_{k}}$, $P^+$, $\lambda^-$, and ${T^*}$, all of which are either known or learned by Algorithm \ref{algor-PI}.
										\end{remark}

\section{Example} \label{sec-simu}
In order to validate the proposed resilient control approach, we consider the adaptive cruise control problem of an autonomous vehicle. 
We aim to keep a desired distance between the autonomous vehicle and a preceding vehicle, i.e., a zero clearance error. 
According to \cite{Li2011CST}, the system matrices are given as follows.
	\begin{align}
		A = & \begin{bmatrix}
			0 & 1 & -\tau_h \\
			0 & 0 & -1 \\
			0 & 0 & -1/T_L 
		\end{bmatrix}, 
		B = \begin{bmatrix}
			0 \\ 0 \\ -K_L/T_L
		\end{bmatrix},  \nonumber \\ 
		D = & \begin{bmatrix}
			0 & 0 & 0\\ 1 & 0 & 1  \\ 0 & 0 & 0
		\end{bmatrix},
		S =  \begin{bmatrix}
		0 & -\omega & 0\\ \omega & 0 & 0 \\ 0 & 0 & 0
		\end{bmatrix},
		\nonumber\\
		C = & \begin{bmatrix}
			1 & 0 & 0
		\end{bmatrix}, 
		F =	\begin{bmatrix}
			0 & 0 & 0
		\end{bmatrix}, \nonumber
	\end{align}
	where 
	the state $x\in\mathbb{R}^3$ consists of clearance error, velocity error and acceleration of the autonomous vehicles, and the input $u\in\mathbb{R}$ is the desired acceleration. Based on the structure of exosystem matrix $S$, we can see that the exostate $v\in\mathbb{R}^3$ includes constant and sinusoidal signals. $Dv$ is the acceleration of the preceding vehicle, implying that its acceleration profile is a combination of constant and sinusoidal signals; see Fig. \ref{fig:headway}a.

	\begin{table}[htbp]
		\centering
		\caption{\label{tab:syspara} System Parameters }
		\begin{tabular}{lll}
			\toprule
			Parameter & Physical Meaning & Value  \\
			\midrule
			$\tau_h[s]$ & Time Headway of Autonomous Vehicle& 1.5 \\
			$T_L[s]$ & Time Constant of Autonomous Vehicle & 0.45 \\
			$K_L$ & System Gain of Autonomous Vehicle & 1 \\ 
			$\omega[Hz]$ & 	Frequency of Exostates & $\pi/3$ \\
			$\eta$ & DoS Frequency Regularization & 1 \\
			$\tau_D[s]$ & DoS Frequency Criterion & 1 \\
			$\kappa[s]$ & DoS Duration Regularization & 0.5 \\
			$T^*$ & Lower Bound of DoS Duration Criterion & 2 \\
			\bottomrule
		\end{tabular}
	\end{table}

	The physical meanings and values of vehicle parameters and DoS-related parameters are illustrated in the Tab. \ref{tab:syspara}. 
	By the collection of online state and input data during $t\in[0,2]s$ when an exploration noise $\xi(t)=0.1\sin(6\pi t)+0.1\sin(12\pi t)$ is introduced to the system input, we implement the PI Algorithm \ref{algor-PI} to learn the resilient optimal control gain and find that the convergence is achieved as the iteration $k=19$. 
	The approximated feedback control gain, feedforward control gain, and value learned by PI are
	\begin{align}
		K=& \begin{bmatrix}
			19.583688, -9.1462677, -4.8707726
		\end{bmatrix}, \nonumber\\
	    P=& \begin{bmatrix}
	    	 185.66245 & -182.53452 & -8.8126596 \\
	    	-182.53452 &  219.34262 &  4.1158205 \\
	    	-8.8126596 &  4.1158205 &  2.191847
	    \end{bmatrix}, \nonumber \\
    	L=& \begin{bmatrix}
    	   -5.8656,   -8.7376,  -19.5893
    	\end{bmatrix}. \nonumber
	\end{align}
	
	For reference, we calculate the resilient optimal control gain, feedforward control gain, and value through solving the ARE (\ref{eqn-ARE2}) and regulator equations (\ref{eqn-regu1})-(\ref{eqn-regu2}) as follows
	\begin{align}
		K^{*}=& \begin{bmatrix}
			19.412827, -9.165516, -4.7660936
		\end{bmatrix}, \nonumber\\
		P^{*}=& \begin{bmatrix}
			183.42892 & -180.67875 & -8.735772 \\
			-180.67875 &  217.68209 & 4.1244822 \\
			 -8.735772 &  4.1244822 & 2.1447421 
		\end{bmatrix}, \nonumber\\
	    L^*=& \begin{bmatrix}
		   -5.8414,   -8.7045,  -19.5144
	\end{bmatrix}. \nonumber
	\end{align}

	{\blue Fig. \ref{fig:K_2024} depicts the difference between the optimal control gain $K^*$ and the learned control gain $K_k$ at each iteration. }	
	We introduce two controllers as a comparison of performance. The first one is an internal-model-based controller \cite{huangjiebook}
	\begin{align} \label{eq-internal}
		   u(t)= &  -K_x x(t)-K_z z(t), \nonumber\\
			\dot z(t)= & S z(t) + G_2 e(t).
	\end{align}
	where $G_2$ is selected such that the pair $(S,G_2)$ is controllable and control gains $K_x$, $K_z$ are chosen such that the matrix $\begin{bmatrix}
		A-BK_x & -BK_z \\ G_2 C & S
	\end{bmatrix}$ is Hurwitz.

	The second one is a variant of model-based resilient controller developed in \cite{FENG2017Auto}
	\begin{align} \label{eq-model-based-resilient}
	\left\{  \begin{matrix}
		u(t) =& -K^* \hat x(t), & \\
		\dot {\hat x}(t) =& A \hat x(t)+Bu(t), & t\in \mathcal{I}_s \\ 
		\hat x(t) =& x(t), & t\notin \mathcal{I}_s.
		\end{matrix}   
	\right.
	\end{align}
    
	We apply the proposed resilient optimal controller, internal-model-based controller, and model-based controller after $t=2s$, respectively.  
	The simulation results are shown in Fig. \ref{fig:headway}b where the shaded area refers to the DoS attack intervals. 
	One can see that the proposed control method outperforms the other two in both transient and steady-state.
	This is because the internal-model-based controller (\ref{eq-internal}) is not aware of DoS attacks, and the state $x(t)$ and tracking error $e(t)$ are not available as feedback when the system is under DoS attacks, although internal model principle is well known to address robust output regulation problems.
	Resilient controller (\ref{eq-model-based-resilient}) has taken the effect of DoS attacks into consideration. However, it can be only ensured that all the signals of the closed-loop system remain bounded for bounded disturbances, which is not enough to achieve asymptotic tracking.

	{\blue We further determine the bound $\tau_D^{19}=14.59s$ based on (\ref{eqn-DoS-T}). }
	However, the DoS frequency criterion $\tau_D=1s$ can be much smaller than the bound in the simulation, which implies that the closed-loop system can defend more frequent attacks than we expect.	

	\begin{figure}[t]
	\centering
	\includegraphics[width=\linewidth]{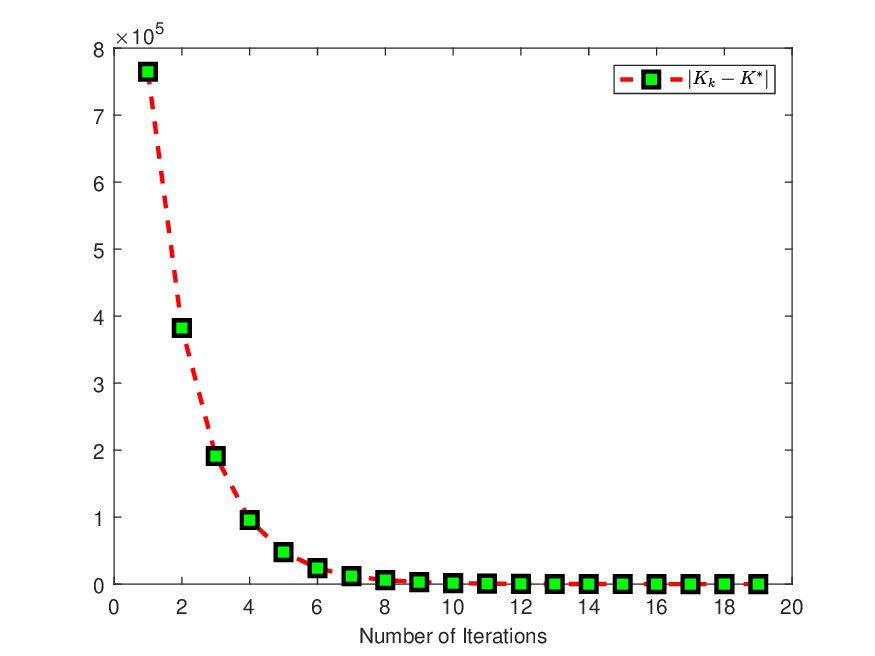}
	\caption{\blue Comparison between the optimal control gain $K^*$ and the learned control gain $K_k$ at each iteration}
	\label{fig:K_2024}
	\end{figure}

	\begin{figure}[t]
		\centering
		\includegraphics[width=\linewidth]{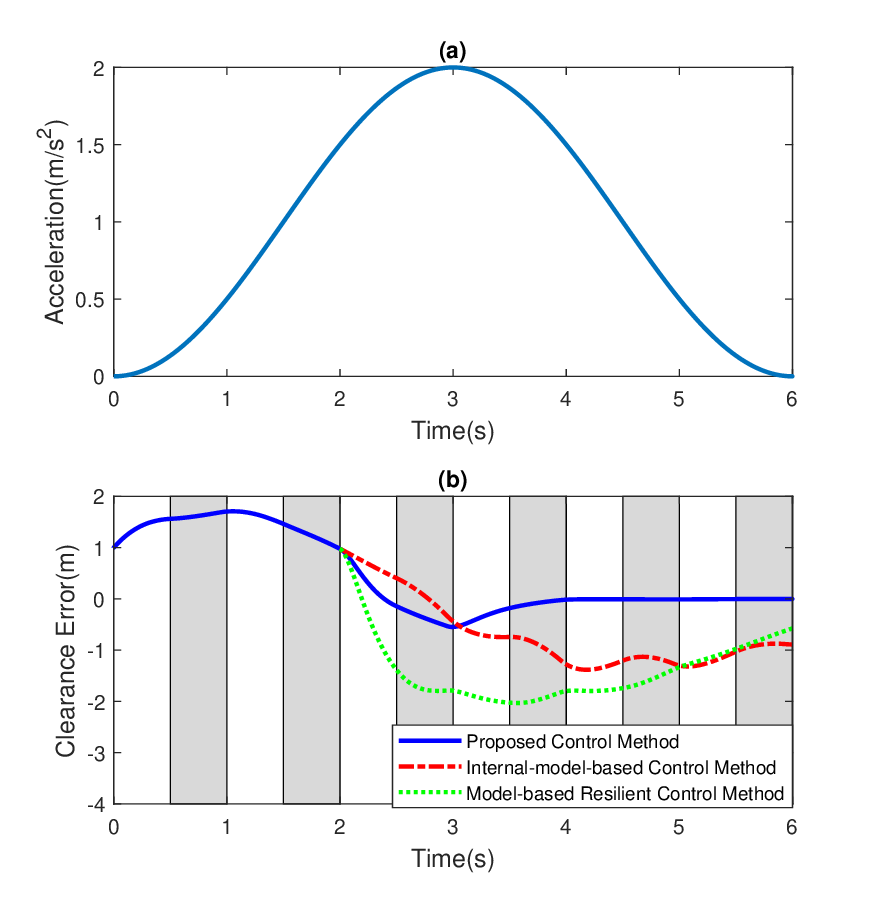}
		\caption{Simulation results: (a) Acceleration profile of preceding vehicle, (b) Clearance error of autonomous vehicle using proposed control method, internal-model-based control method, and model-based resilient control method. }
		\label{fig:headway}
	\end{figure}

	\section{Conclusion} \label{sec-conclusion}
	
	In this paper, a learning-based resilient optimal control framework has been constructed using recent developments in reinforcement learning and ADP. 
	Under this framework, we first estimate the divergence rate of the system under 
	DoS attacks. Based on the estimated rate and a lower bound of DoS
	duration, we determine a desired convergence rate for the closed-loop
	system with normal communication. 
	Then, we apply ADP methods to learn a resilient optimal control policy under the desired rate.
	Finally, through rigorous theoretical analysis and numerical simulations, we have shown that the control system has improved resilience against cyberattacks with the learned adaptive controller.

	\bibliographystyle{IEEETranS}
	\bibliography{myreference}           
	
\end{document}